\documentclass[a4paper,USenglish,autoref]{lipics-v2021}

\usepackage{algorithm}
\usepackage[noend]{algpseudocode}
\usepackage{booktabs}
\usepackage{cite}
\usepackage{mathrsfs}

\newcommand{\bigo}[1]{\ensuremath{\mathcal{O}(#1)}}

\newcommand{\true}{\ensuremath{\textsc{true}}}
\newcommand{\false}{\ensuremath{\textsc{false}}}

\newcommand{\alg}{\ensuremath{\mathcal{A}}}

\newcommand{\lbl}{\ensuremath{\ell}}
\newcommand{\disconnectset}{\ensuremath{D}}
\newcommand{\lock}{\ensuremath{\texttt{lock}}}
\newcommand{\state}{\ensuremath{\texttt{state}}}
\newcommand{\phase}{\ensuremath{\texttt{phase}}}
\newcommand{\lockset}{\ensuremath{L}}
\newcommand{\replyset}{\ensuremath{R}}
\newcommand{\winset}{\ensuremath{W}}
\newcommand{\holdset}{\ensuremath{H}}
\newcommand{\applyset}{\ensuremath{A}}
\newcommand{\candidateset}{\ensuremath{C}}
\newcommand{\compete}{\ensuremath{P}}
\newcommand{\msg}[2]{\ensuremath{\texttt{#1(#2)}}}

\newcommand{\Lock}{\ensuremath{\textsc{Lock}}}
\newcommand{\Unlock}{\ensuremath{\textsc{Unlock}}}

\makeatletter
\algrenewcommand\ALG@beginalgorithmic{\small}
\algrenewcommand\alglinenumber[1]{\scriptsize #1:}
\makeatother

\algdef{SE}[SUBALG]{Indent}{EndIndent}{}{\algorithmicend\ }%
\algtext*{Indent}
\algtext*{EndIndent}
\makeatletter
\newcommand{\multiline}[1]{%
  \begin{tabularx}{\dimexpr\linewidth-\ALG@thistlm}[t]{@{}X@{}}
    #1
  \end{tabularx}
}
\makeatother

\newif\ifarxiv
\arxivtrue

\ifarxiv
\hideLIPIcs
\else\fi

\newcommand{\junk}[1]{}

\newif\iffigabbrv
\figabbrvfalse   

\title{On the Runtime of Local Mutual Exclusion for Anonymous Dynamic Networks}
\titlerunning{On the Runtime of Local Mutual Exclusion for Anonymous Dynamic Networks}
\author{Anya Chaturvedi}{School of Computing and Augmented Intelligence\\Arizona State University, Tempe, AZ, USA}{anya.chaturvedi@asu.edu}{https://orcid.org/0000-0002-4071-0467}{NSF (CCF-2312537) and U.S.\ ARO (MURI W911NF-19-1-0233).}
\author{Joshua J. Daymude}{School of Computing and Augmented Intelligence \and Biodesign Center for Biocomputing, Security and Society\\Arizona State University, Tempe, AZ, USA}{jdaymude@asu.edu}{https://orcid.org/0000-0001-7294-5626}{NSF (CCF-2312537).}
\author{Andr\'ea W. Richa}{School of Computing and Augmented Intelligence \and Biodesign Center for Biocomputing, Security and Society\\Arizona State University, Tempe, AZ, USA}{aricha@asu.edu}{https://orcid.org/0000-0003-3592-3756}{NSF (CCF-2312537, CCF-2106917) and U.S.\ ARO (MURI W911NF-19-1-0233).}
\authorrunning{A.\ Chaturvedi, J.\ J.\ Daymude, and A.\ W.\ Richa}

\Copyright{Anya Chaturvedi, Joshua J.\ Daymude, and Andr\'ea W.\ Richa}



\nolinenumbers

\ifarxiv\else
\EventEditors{}
\EventNoEds{0}
\EventLongTitle{}
\EventShortTitle{}
\EventAcronym{SAND}
\EventYear{2025}
\EventDate{}
\EventLocation{}
\EventLogo{}
\SeriesVolume{330}
\ArticleNo{15}
\fi

\begin{document}

\maketitle

\begin{abstract}
    Algorithms for \textit{mutual exclusion} aim to isolate potentially concurrent accesses to the same shared resources.
    Motivated by distributed computing research on programmable matter and population protocols where interactions among entities are often assumed to be isolated, Daymude, Richa, and Scheideler (SAND`22) introduced a variant of the \textit{local mutual exclusion} problem that applies to arbitrary dynamic networks: each node, on issuing a lock request, must acquire exclusive locks on itself and all its \textit{persistent neighbors}, i.e., the neighbors that remain connected to it over the duration of the lock request.
    Assuming adversarial edge dynamics, semi-synchronous or asynchronous concurrency, and anonymous nodes communicating via message passing, their randomized algorithm achieves mutual exclusion (non-intersecting lock sets) and lockout freedom (eventual success 
    with probability 1).
    However, they did not analyze their algorithm's runtime.
    In this paper, we prove that any node will successfully lock itself and its persistent neighbors within $\bigo{n\Delta^3}$ open rounds of its lock request in expectation, where $n$ is the number of nodes in the dynamic network, $\Delta$ is the maximum degree of the dynamic network, rounds are normalized to the execution time of the ``slowest'' node, and ``closed'' rounds when some persistent neighbors are already locked by another node are ignored (i.e., only ``open" rounds are considered).
\end{abstract}

\section{Introduction} \label{sec:intro}

The algorithmic theory of \textit{dynamic networks} formally characterizes distributed systems of processes (or \textit{nodes}) whose connections change---periodically, randomly, or even adversarially---over time.
Using models like \textit{time-varying graphs} (TVGs)~\cite{Casteigts2018-journeythrough,Casteigts2012-timevaryinggraphs}, theoreticians have successfully obtained algorithms for many core distributed computing problems including broadcast~\cite{Casteigts2010-deterministiccomputations,Casteigts2015-shortestfastest,Raynal2014-simplebroadcast,Parzych2024-memorylower}, consensus and agreement~\cite{Coulouma2013-characterizationdynamic,Biely2012-agreementdirected,Kuhn2011-coordinatedconsensus}, leader election~\cite{Ingram2009-asynchronousleader,Augustine2015-fastbyzantine}, counting~\cite{Michail2013-namingcounting,Chakraborty2018-fasterexactcounting,DiLuna2016-nontrivial,Kowalski2020-polynomialcounting,DiLuna2023-optimalcomputation}, and more (see~\cite{Casteigts2018-journeythrough,Augustine2016-distributedalgorithmic} for surveys).
However---with relatively few exceptions~\cite{Casteigts2010-deterministiccomputations,Dubois2015-enablingminimal,Flocchini2012-searchingblack,Ingram2009-asynchronousleader}---an overwhelming majority of these algorithms assume \textit{synchronous concurrency} in which algorithm executions alternate between stages of node computation and instants of network dynamics.

Towards a deeper understanding of concurrency control in dynamic networks, Daymude, Richa, and Scheideler proposed a variant of the \textit{local mutual exclusion} problem  in \textit{dynamic anonymous networks}~\cite{Daymude2022-localmutual} in which anonymous nodes of a dynamic network must lock themselves and their \textit{persistent neighbors} (i.e., those that remain connected throughout the duration of the lock request) before entering their ``critical sections'' (i.e., before doing an actual action execution).
Their randomized algorithm for \Lock\ and \Unlock\ primitives achieves mutual exclusion (non-intersecting lock sets) and lockout freedom (every lock request eventually succeeds with probability 1) even when nodes are anonymous, operate under asynchronous concurrency, and experience edge dynamics concurrently with their computation.
Thus, distributed algorithms for dynamic networks can first be designed and analyzed for the sequential setting in which at most one node is active per time, and then---using this locking mechanism to isolate concurrently executed actions and dynamics---can be directly simulated in the asynchronous setting.

However, that algorithm is missing runtime analysis, obscuring the overhead that other algorithms would incur by using this lock-based concurrency control. 
In this paper, we address this gap, proving an expectation bound on the time any node waits from issuing its lock request to obtaining locks on itself and its persistent neighbors.
This analysis requires minor but nontrivial refinements to both the model (specifically, the model is now described in terms of 
adversarial fairness over node activations, 
rather than 
over individual enabled actions) 
and the algorithm (to close loopholes that allow an adversary to achieve arbitrarily large but finite runtimes), so we also reprove the algorithm's mutual exclusion and lockout freedom guarantees w.r.t.\ these changes.

\subparagraph*{Our Contributions.}
We summarize our main results as follows:
\begin{itemize}
    \item 
    To support our runtime analysis, we slightly refine
    the assumptions on adversarial fairness and parts of the local mutual exclusion algorithm of~\cite{Daymude2022-localmutual}. We 
    then 
    prove that the  refined
    algorithm still satisfies \textit{mutual exclusion} (non-intersecting lock sets) 
    and \textit{lockout freedom} (eventual success with probability 1).

    \item We prove that, using this (refined)
    local mutual exclusion algorithm, any node issuing a lock request will obtain locks on itself and its persistent neighbors within $\bigo{n\Delta^3}$ \textit{open rounds} in expectation, where $n$ is the number of nodes in the dynamic network, $\Delta$ is the maximum degree of the dynamic network, rounds can be viewed as
    normalized to the execution time of the ``slowest'' node, and open rounds are those in which none of the node's persistent neighbors are already locked by some other node, prohibiting progress.
\end{itemize}

\subsection{Related Work} \label{subsec:relwork}

Since its introduction~\cite{Dijkstra1965-mutualexclusion}, the \textit{mutual exclusion} problem has been one of the most widely-studied problems in parallel and distributed computing.
Here, we briefly situate the variant of the \textit{local mutual exclusion} problem introduced in~\cite{Daymude2022-localmutual} and refer the interested reader to their related work section for a detailed literature review.
Many past works solve mutual exclusion or a related problem for message passing systems on static network topologies.
For example, the \textit{arrow protocol}~\cite{Demmer1998-arrowdistributed,Raymond1989-treemutex} and its recent improvements~\cite{Ghodselahi2017-dynamicanalysis,Khanchandani2019-arvydistributed} use a token-passing scheme that requires only constant memory per node, and more recent results address mutual exclusion directly for fully anonymous systems~\cite{Raynal2020-mutualexclusion}.
However, these do not immediately apply to dynamic networks.
Mutual exclusion and contention resolution algorithms for \textit{mobile ad hoc networks} (MANETs) handle node and edge dynamics but typically assume time-ordered, instantaneous transmissions that are more powerful than our communication model of asynchronous message passing~\cite{Attiya2010-efficientrobust,Baldoni2002-distributedmutual,Benchaiba2004-distributedmutual,Chen2005-selfstabilizingdynamic,Sharma2014-tokenbased,Walter2001-mutualexclusion,Bender2005-adversarialcontention,Cali2000-ieee80211,Capetanakis1979-treealgorithms}.
Related works on \textit{self-stabilizing overlay networks} also consider edge dynamics but assume nodes have identifiers~\cite{Feldmann2020-surveyalgorithms}.

In contrast, the local mutual exclusion problem variant and probabilistic algorithm of~\cite{Daymude2022-localmutual} explicitly consider the challenge of neighborhood mutual exclusion on a \textit{time-varying graph}~\cite{Casteigts2018-journeythrough,Casteigts2012-timevaryinggraphs} with adversarial edge dynamics where nodes are anonymous and operate under semi-synchronous or asynchronous concurrency.
Since our primary goal is to analyze the locking time of this algorithm, we use the same computational model (Section~\ref{subsec:model}), problem statement (Section~\ref{subsec:problem}), and algorithm (Section~\ref{sec:alg})---with some key minor modifications that we will highlight as we go.
These changes support the clarity and tractability of our runtime analysis but do not fundamentally affect the strength or correctness of the results in~\cite{Daymude2022-localmutual}.

\subsection{Computational Model} \label{subsec:model}

Following the model in~\cite{Daymude2022-localmutual}, we consider a distributed system comprising a fixed set of nodes $V$.
Each node is assumed to be \textit{anonymous}, lacking a unique identifier, and has a local memory storing its \textit{state}.
Nodes communicate with each other via message passing over a communication graph whose topology changes over time.
We model this topology using a \textit{time-varying graph} $\mathcal{G} = (V, E, T, \rho)$ where $V$ is the set of nodes, $E$ is a (static) set of undirected pairwise edges between nodes, $T = \mathbb{N}$
 
is the \textit{lifetime} of $\mathcal{G}$, and $\rho : E \times T \to \{0,1\}$ is the \textit{presence} function indicating whether or not a given edge exists at a given time~\cite{Casteigts2012-timevaryinggraphs}.
A \textit{snapshot} of $\mathcal{G}$ at time $t \in T$ is the undirected graph $G_t = (V, \{e \in E : \rho(e, t) = 1\})$ and the \textit{neighborhood} of a node $u \in V$ at time $t \in T$ is the set $N_t(u) = \{v \in V : \rho(\{u, v\}, t) = 1\}$.
For $i \geq 0$, the $i$-th \textit{stage} lasts from time $i$ to the instant just before time $i+1$; thus, the communication graph in stage $i$ is $G_i$.

We assume that an adversary controls the presence function $\rho$ and that $E$ is the complete set of edges on nodes $V$; i.e., we do not limit which edges the adversary can introduce.
The only constraint we place on the adversary's topological changes is $\forall t \in T, u \in V, |N_t(u)| \leq \Delta$, where $\Delta > 0$ is the fixed number of \textit{ports} per node.
When the adversary establishes a new connection between nodes $u$ and $v$, it must assign the endpoints of edge $\{u, v\}$ to open ports on $u$ and $v$ (and cannot do so if either node has no open ports).
Node $u$ locally identifies $\{u, v\}$ using its corresponding \textit{port label} $\lbl \in \{1, \ldots, \Delta\}$ and $v$ does likewise.
For convenience of notation, we use $\lbl_u(v)$ to refer to the label of the port on node $u$ that is assigned to the edge $\{u, v\}$; this mapping of port labels to nodes is not available to the nodes.
Edge endpoints remain in their assigned ports (and thus labels remain fixed) until disconnection, but nodes $u$ and $v$ may label $\{u, v\}$ differently and their labels are not known to each other a priori.
Each node has a \textit{disconnection detector} that adds the label of any port whose connection is severed to a set $\disconnectset \subseteq \{1, \ldots, \Delta\}$.
A node's disconnection detector provides it a snapshot of $\disconnectset$ whenever it starts an action execution (see below) and then resets $\disconnectset$ to $\emptyset$.\footnote{Without this assumption, the adversary could disconnect an edge assigned to port $\lbl$ of node $u$ and then immediately connect a different edge to $\lbl$, causing an indistinguishability issue for node $u$.}

Nodes communicate via message passing.
A node $u$ sends a message $m$ to a neighbor $v$ by calling $\textsc{Send}(m, \lbl_u(v))$.
Message $m$ remains in transit until either $v$ receives and processes $m$ at a later stage chosen by the adversary, or $u$ and $v$ are disconnected and $m$ is lost.
Multiple messages in transit from $u$ to $v$ may be received by $v$ in a different order than they were sent.
A node always knows from which port it received a given message.

All nodes execute the same distributed algorithm $\alg$, which is a set of \textit{actions} each of the form $\langle label\rangle: \langle guard\rangle \to \langle operations\rangle$.
An action's \textit{label} specifies its name.
Its \textit{guard} is a Boolean predicate determining whether a node $u$ can execute it based on the state of $u$ and any message in transit that $u$ may receive.
An action is \textit{enabled} for a node $u$ if its guard is \true\ for $u$; a node $u$ is \textit{enabled} if it has at least one enabled action.
An action's \textit{operations} specify what a node does when executing the action, structured as (\textit{i}) receiving at most one message chosen by the adversary, (\textit{ii}) a finite amount of internal computation and state updates, and (\textit{iii}) at most one call to $\textsc{Send}(m, \ell)$ per port label $\ell$.

Each node executes its own instance of $\alg$ independently, sequentially (executing at most one action at a time), and reliably (meaning we do not consider crash or Byzantine faults).
We assume an adversarial scheduler controls the timing of node activations.
We primarily focus on \textit{semi-synchronous} concurrency: in each stage, the adversary activates any (possibly empty) subset of enabled nodes concurrently and all resulting action executions complete within that stage.
\textit{Asynchronous} concurrency allows action executions to span arbitrary finite time intervals.

In~\cite{Daymude2022-localmutual}, the adversary (scheduler) 
had a notion of fairness that ensured every continuously enabled action is eventually executed and every message in transit on a continuously existent edge is eventually processed.
Here, we reinterpret this fairness assumption to support our runtime analysis in terms of node activations.
From the perspective of this adversary, each enabled node $u$ has a nonempty set of \textit{enabled action executions} $\{(\alpha, m)\}$ it could enact, where $\alpha$ is an action enabled for $u$ and---if $\alpha$ is enabled by one or more messages in transit to $u$---$m$ is any one of these messages.
An enabled action execution $(\alpha, m)$ is disabled for a node $u$ when $\alpha$ is disabled for $u$ or $m$ is lost due to disconnection.
This adversary must activate enabled nodes such that any continuously enabled node is eventually executed, but once the adversary activates an enabled node, one of that node's enabled action executions is enacted \textit{uniformly at random}.\footnote{Following the taxonomy in~\cite{dubois2011taxonomydaemonsselfstabilization}, the scheduler adversary would be (weakly) fair and have no restrictions on distribution, boundedeness nor enabledness.} We discuss possible alternative approaches for the selection of the enabled action executions by a node, including that of (deterministic) FIFO queues at each node, in Section~\ref{sec:conclude}.

Finally, to bound the locking time of the local mutual exclusion algorithm, we define \textit{rounds}, which can be informally viewed as 
normalized to the ``slowest'' node's action execution.
Let $t_r \geq 0$ be the time that round $r \in \{0, 1, \ldots\}$ starts, with $t_0 = 0$.
Define $\mathcal{E}_r$ as the set of nodes that are enabled or currently executing an action at time $t_r$.
Then $t_{r+1} > t_r$ is the earliest time by which every node in $\mathcal{E}_r$ has either been disabled or completed an action execution.
Hence, round $r$ is the time interval $[t_r, t_{r+1})$; note that one round may span multiple stages.

\subsection{Local Mutual Exclusion} \label{subsec:problem}

Like the classical mutual exclusion problem~\cite{Dijkstra1965-mutualexclusion},  the \textit{local mutual exclusion} variant considered in~\cite{Daymude2022-localmutual} is defined w.r.t.\ a pair of primitives \Lock\ and \Unlock.
A node issues a lock request by calling \Lock; once acquired, it is assumed that a node eventually releases these locks by calling \Unlock.
Formally, each node $u$ stores a variable $\lock \in \{\bot, 0, \ldots, \Delta\}$ that is equal to $\bot$ if $u$ is unlocked, $0$ if $u$ has locked itself, and $\lbl_u(v) \in \{1, \ldots, \Delta\}$ if $u$ is locked by $v$.
The \textit{lock set} of a node $u$ in stage $i$ is $\mathcal{L}_i(u) = \{v \in N_i(u) : \lock(v) = \ell_v(u)\}$ 
which additionally includes $u$ itself if $\lock(u) = 0$.
Suppose that in stage $i$, a node $u$ calls \Lock\ to issue a lock request of its current closed neighborhood $N_i[u] = \{u\} \cup N_i(u)$.
This lock request \textit{succeeds} at some later stage $j > i$ if stage $j$ is the first in which $\mathcal{L}_j(u) = \{u\} \cup \{v \in N_i(u) : \forall t \in [i, j], \{u, v\} \in G_t\}$; i.e., $j$ is the earliest stage in which $u$ obtains locks for itself and every \textit{persistent} neighbor that remained connected to $u$ in stages $i$ through $j$.
In our context, a (randomized) algorithm $\alg$ \textit{solves} local mutual exclusion if it implements \Lock\ and \Unlock\ while satisfying the following properties:

\begin{itemize}
    \item \textit{Mutual Exclusion.} For all stages $i \in T$ and all pairs of nodes $u, v \in V$, $\mathcal{L}_i(u) \cap \mathcal{L}_i(v) = \emptyset$.
    
    \item \textit{Lockout Freedom.} Every issued lock request eventually succeeds with probability 1.
\end{itemize}

Since each node's \lock\ variable points to at most one node per time, it is impossible for two nodes' lock sets to intersect, trivially satisfying the mutual exclusion property.
But lockout freedom is more elusive, especially in the dynamic setting.

\section{Algorithm for Local Mutual Exclusion} \label{sec:alg}

The randomized algorithm for local mutual exclusion given by Daymude, Richa and Scheideler~\cite{Daymude2022-localmutual}  specifies actions for the \Lock\ and \Unlock\ primitives satisfying mutual exclusion and lockout freedom.
In order for our runtime analysis of that algorithm (Section~\ref{sec:analysis}) to be coherent and self-contained, we reproduce their algorithm's narrative description and pseudocode and highlight our minor modifications where relevant.

An execution of the \Lock\ primitive by a node $u$ is organized into two phases: a \textit{preparation phase} (Algorithm~\ref{alg:lockprep}) in which $u$ determines and notifies the nodes $L(u)$ it intends to lock, and a \textit{competition phase} (Algorithm~\ref{alg:lockcompete}) in which $u$ attempts to lock all nodes in $L(u)$, contending with any other node $v$ for which $L(v) \cap L(u) \neq \emptyset$.
An execution of the \Unlock\ primitive (Algorithm~\ref{alg:unlock}) by node $u$ is straightforward, simply notifying all nodes in $L(u)$ that their locks are released.
All local variables used in this algorithm are listed in Table~\ref{tab:variables} as they appear in the pseudocode.
In a slight abuse of notation, $N[u]$ and the subsets thereof represent both the nodes in the closed neighborhood of $u$ and the port labels of $u$ they are connected to.
For clarity of presentation, the algorithm pseudocode allows a node to send messages to itself (via ``port 0'') just as it sends messages to its neighbors, though in reality these self-messages would be implemented with in-memory variable updates.

\begin{table}[t]
    \centering
    \caption{The notation, domain, initialization, and description of the local variables used in the algorithm for local mutual exclusion by a node $u$.}
    \label{tab:variables}
    \begin{tabular}{cp{99pt}cp{208pt}}
        \toprule
        \textbf{Var.} & \textbf{Domain} & \textbf{Init.} & \textbf{Description} \\
        \midrule
        \lock & $\{\bot, 0, \ldots, \Delta\}$ & $\bot$ & $\bot$ if $u$ is unlocked, $0$ if $u$ has locked itself, and $\lbl_u(v)$ if $u$ is locked by $v$ \\
        \state & $\{\bot, \textsc{prepare}, \textsc{compete}$, $\textsc{win}, \textsc{locked}, \textsc{unlock}\}$ & $\bot$ & The lock state of node $u$ \\
        \phase & $\{\bot, \textsc{prepare}, \textsc{compete}\}$ & $\bot$ & The algorithm phase node $u$ is in \\
        \lockset & $\subseteq N[u]$ & $\emptyset$ & Ports (nodes) $u$ intends to lock \\
        \replyset & $\subseteq N[u]$ & $\emptyset$ & Ports via which $u$ has received \msg{ready}{}, \msg{ack-lock}{}, or \msg{ack-unlock}{} responses \\
        $\winset$ & $\subseteq N[u] \times \{\true, \false\}$ & $\emptyset$ & Port-outcome pairs of \msg{win}{} messages $u$ has received \\
        \holdset & $\subseteq N[u]$ & $\emptyset$ & Ports (nodes) on hold for the competition to lock $u$ \\
        \applyset & $\subseteq N[u]$ & $\emptyset$ & Ports (nodes) of applicants that can join the competition to lock $u$ \\
        \candidateset & $\subseteq N[u]$ & $\emptyset$ & Ports (nodes) of candidates competing to lock $u$ \\
        \compete & $\subseteq C(u) \times \{0, \ldots, K-1\}$ & $\emptyset$ & Port-priority pairs of the candidates \\
        \disconnectset & $\subseteq N[u]$ & $\emptyset$ & Ports that were disconnected since last execution \\
        \bottomrule
    \end{tabular}
\end{table}

\begin{algorithm}[t]
\caption{The \Lock\ Operation: Preparation Phase for Node $u$} \label{alg:lockprep}
\begin{algorithmic}[1]
    \State \textsc{InitLock$^{I}$}: On \Lock\ being called $\to$  \Comment{Initiator initiates a lock request.}
        \Indent
            \If {$\state = \bot$} \Comment{Only one locking operation at a time.}
                \State \Call{CleanUp}{\,}.
                \State Set $\state \gets \textsc{prepare}$ and $\lockset \gets N[u]$. 
                \ForAll{$\lbl \in \lockset$} \Call{Send}{\msg{prepare}{}, $\lbl$}.
                \EndFor
            \EndIf
        \EndIndent
    \State \textsc{ReceivePrepare$^{P}$}: On receiving \msg{prepare}{} via port $\lbl$ $\to$
        \Indent
            \State \Call{CleanUp}{\,}.
            \If {$\phase = \textsc{compete}$} set $\holdset \gets \holdset \cup \{\lbl\}$. \Comment{Put $\lbl$ on hold if already competing.}
            \Else {}
                \State Set $\applyset \gets \applyset \cup \{\lbl\}$ and $\phase \gets \textsc{prepare}$. \Comment{Add $\lbl$ as an applicant otherwise.}
                \State \Call{Send}{\msg{ready}{}, $\lbl$}.
            \EndIf
        \EndIndent     
    \State \textsc{ReceiveReady$^{I}$}: On receiving \msg{ready}{} via port $\lbl$ $\to$
        \Indent
            \State \Call{CleanUp}{\,}.
            \State Set $\replyset \gets \replyset \cup \{\lbl\}$.
        \EndIndent     
    \State \textsc{CheckStart$^{I}$}: $(\state = \textsc{prepare}) \wedge (\replyset \setminus \disconnectset = \lockset \setminus \disconnectset)$ $\to$  \Comment{All \msg{ready}{} messages received.}
        \Indent
            \State \Call{CleanUp}{\,}.
            \State Set $\state \gets \textsc{compete}$, $\replyset \gets \emptyset$, and $\winset \gets \emptyset$.
            \State Choose priority $p \in \{0, \ldots, K-1\}$ uniformly at random.
            \ForAll{$\lbl \in \lockset$} \Call{Send}{\msg{request-lock}{$p$}, $\lbl$}.
            \EndFor
        \EndIndent
    \Function {CleanUp}{\,}  \Comment{Helper function for processing disconnections $\disconnectset$.}
        \ForAll{$\lbl \in \disconnectset$}
            \If {$\lock = \lbl$} $\lock \gets \bot$.
            \EndIf
            \State Remove $\lbl$ from all sets: $\lockset \gets \lockset \setminus \{\lbl\}$, $\replyset \gets \replyset \setminus \{\lbl\}$, $\winset \gets \winset \setminus \{(\lbl, \cdot)\}$, $\holdset \gets \holdset \setminus \{\lbl\}$,
            \State $\applyset \gets \applyset \setminus \{\lbl\}$, $\candidateset \gets \candidateset \setminus \{\lbl\}$, and $\compete \gets \compete \setminus \{(\lbl, \cdot)\}$.
        \EndFor
        \If {$\candidateset = \emptyset$} 
            \ForAll {$\lbl \in \holdset$} \Call{Send}{\msg{ready}{}, $\lbl$}.
            \EndFor
            \State Set $\applyset \gets \applyset \cup \holdset$ and $\holdset \gets \emptyset$. \Comment{All nodes on hold become applicants.}
            \If {$\applyset \neq \emptyset$} set $\phase \gets \textsc{prepare}$.
            \Else {} set $\phase \gets \bot$.
            \EndIf
        \EndIf
    \EndFunction
\end{algorithmic}
\end{algorithm}

\begin{algorithm}[t]
\caption{The \Lock\ Operation: Competition Phase for Node $u$} \label{alg:lockcompete}
\begin{algorithmic}[1]
    \State \textsc{ReceiveRequest$^{P}$}: On receiving \msg{request-lock}{$p$} via port $\lbl$ $\to$
        \Indent
            \State \Call{CleanUp}{\,}.
            \If {$\lbl \in \applyset$} set $\applyset \gets \applyset \setminus \{\lbl\}$ and $\candidateset \gets \candidateset \cup \{\lbl\}$.
            \EndIf
            \State Set $\compete \gets \compete \cup \{(\lbl, p)\}$ and $\phase \gets \textsc{compete}$. \Comment{Close competition.}
         \EndIndent
    \State \textsc{CheckPriorities$^{P}$}: $(\phase = \textsc{compete}) \wedge (|\candidateset \setminus \disconnectset| = |\compete \setminus \disconnectset|) \wedge (\applyset \setminus \disconnectset = \emptyset)$ $\to$ 
        \Indent
            \State \Call{CleanUp}{\,}. \Comment{All priorities received.}
            \If {$\lock = \bot$ and $\exists (\lbl, p) \in \compete$ with a unique highest $p$}
                \State \Call{Send}{\msg{win}{\true}, $\lbl$} and \Call{Send}{\msg{win}{\false}, $\lbl'$} for all $\lbl' \in \candidateset \setminus \{\lbl\}$.
            \Else {} \Call{Send}{\msg{win}{\false}, $\lbl$} for all $\lbl \in \candidateset$.
            \EndIf
            \State Reset $\compete \gets \emptyset$.  \Comment{Competition is over.}
        \EndIndent
    \State \textsc{ReceiveWin$^{I}$}: On receiving \msg{win}{$b$} via port $\lbl$ $\to$
        \Indent
            \State \Call{CleanUp}{\,}.
            \State Set $\winset \gets \winset \cup \{(\lbl, b)\}$.
        \EndIndent
     \State \textsc{CheckWin$^{I}$}: $(\state = \textsc{compete}) \wedge (|\winset \setminus \disconnectset| = |\lockset \setminus \disconnectset|)$ $\to$ \Comment{All \msg{win}{$b$} replies received.}
        \Indent
            \State \Call{CleanUp}{\,}.
            \If {$\exists(\cdot, \false) \in \winset$}  \Comment{Start new locking attempt.}
                \State Choose priority $p \in \{0, \ldots, K-1\}$ uniformly at random.
                \ForAll{$\lbl \in \lockset$} \Call{Send}{\msg{request-lock}{$p$}, $\lbl$}.
                \EndFor
            \Else {}  \Comment{Succeeded in locking.}
                \State Set $\state \gets \textsc{win}$ and reset $\replyset \gets \emptyset$.
                \ForAll{$\lbl \in \lockset$} \Call{Send}{\msg{set-lock}{}, $\lbl$}.
                \EndFor
            \EndIf
            \State Reset $\winset \gets \emptyset$.
        \EndIndent
    \State \textsc{ReceiveSetLock$^{P}$}: On receiving \msg{set-lock}{} via port $\lbl$ $\to$
        \Indent
            \State Set $\lock \gets \lbl$ and $\candidateset \gets \candidateset \setminus \{\lbl\}$.
            \State \Call{CleanUp}{\,}.
            \State \Call{Send}{\msg{ack-lock}{}, $\lbl$}.
        \EndIndent
    \State \textsc{ReceiveAckLock$^{I}$}: On receiving \msg{ack-lock}{} via port $\lbl$ $\to$
        \Indent
            \State \Call{CleanUp}{\,}.
            \State Set $\replyset \gets \replyset \cup \{\lbl\}$.
        \EndIndent
    \State \textsc{CheckDone$^{I}$}: $(\state = \textsc{win}) \wedge (\replyset \setminus \disconnectset= \lockset \setminus \disconnectset)$ $\to$  \Comment{All lock acknowledgements received.}
        \Indent
            \State \Call{CleanUp}{\,}.
            \State Set $\state \gets \textsc{locked}$ and reset $\replyset = \emptyset$.
            \State \Return $\lockset$. \Comment{Locking complete.}
        \EndIndent
\end{algorithmic}
\end{algorithm}

\begin{algorithm}[t]
\caption{The $\Unlock$ Operation for Node $u$} \label{alg:unlock}
\begin{algorithmic}[1]
    \State \textsc{InitUnlock$^{I}$}: On \Unlock\ being called $\to$  \Comment{Initiator initiates an unlock.}
    \Indent
        \If {$\state = \textsc{locked}$} \Comment{Only one \Unlock\ per successful \Lock.}
            \State \Call{CleanUp}{\,}.
            \State Set $\state \gets \textsc{unlock}$ and reset $\replyset \gets \emptyset$.
            \ForAll {$\lbl \in \lockset$} \Call{Send}{\msg{release-lock}{}, $\lbl$}.
            \EndFor
        \EndIf
    \EndIndent
    \State \textsc{ReceiveRelease$^{P}$}: On receiving \msg{release-lock}{} via port $\lbl$ $\to$
       \Indent
          \State \Call{CleanUp}{\,}.
          \State Set $\lock \gets \bot$ and \Call{Send}{\msg{ack-unlock}{}, $\lbl$}.
       \EndIndent
    \State \textsc{ReceiveAckUnlock$^{I}$}: On receiving \msg{ack-unlock}{} via port $\lbl$ $\to$
       \Indent
          \State \Call{CleanUp}{\,}.
          \State Set $\replyset \gets \replyset \cup \{\lbl\}$.
       \EndIndent
    \State \textsc{CheckUnlocked$^{I}$}: $(\state = \textsc{unlock}) \wedge (\replyset \setminus \disconnectset = \lockset \setminus \disconnectset)$ $\to$  \Comment{All unlock acknowledgements received.}
        \Indent
            \State \Call{CleanUp}{\,}.
            \State Reset $\state \gets \bot$ and $\replyset = \emptyset$. \Comment{Unlocking complete.}
        \EndIndent
\end{algorithmic}
\end{algorithm}

Nodes that issue lock requests are called \textit{initiators} and nodes that are being locked or unlocked are \textit{participants}; it is possible for a node to be an initiator and participant simultaneously.
Initiators progress through a series of \textit{lock states} associated with the \state\ variable; participants advance through the algorithm's \textit{phases} as indicated by the \phase\ variable.
We first describe the algorithm from an initiator's perspective and then describe the complementary participants' actions.
In the pseudocode, all actions executed by initiators (resp., participants) are marked with an $I$ (resp., $P$) superscript.

A special \textsc{CleanUp} helper function runs at the beginning of every action execution to ensure that nodes adapt to any disconnections affecting their variables that may have occurred since they last acted (already reflected in the current guard evaluations, which all disregard ports in $D$), so we omit the handling of these disconnections in the description in the next paragraphs.
Note that this is a small modification of the original algorithm in~\cite{Daymude2022-localmutual}, where \textsc{CleanUp} appeared as both a helper function and its own action; we eliminated the standalone action in this version of the algorithm because, in the context of runtime analysis, the adversary could abuse this standalone action to drive up the number of rounds without making any real progress.\footnote{In addition to the changes with respect to the \textsc{CleanUp} function, a typo in the guard of the pseudocode for {\sc CheckPriorities$^P$} was corrected in this version of the algorithm.} 

When an initiator $u$ calls \Lock, it advances to the \textsc{prepare} state, sets $\lockset(u)$ to all nodes in its closed neighborhood $N[u]$, and then sends \msg{prepare}{} messages to all nodes of $\lockset(u)$.
Once it has received \msg{ready}{} responses from all nodes of $\lockset(u)$, it advances to the \textsc{compete} state and joins the competitions for each node in $\lockset(u)$ by sending \msg{request-lock}{$p$} messages to all nodes of $\lockset(u)$, where $p$ is a priority chosen uniformly at random from $\{0, \ldots, K - 1\}$, where $K = \Theta(\Delta^2)$. 
It then waits for the outcomes of these competitions.
If it receives at least one \msg{win}{\false} message, it lost this competition and must compete again.
Otherwise, if all responses are \msg{win}{\true}, it advances to the \textsc{win} state and sends \msg{set-lock}{} messages to all nodes of $\lockset(u)$.
Once it has received \msg{ack-lock}{} responses from all nodes of $\lockset(u)$, it advances to the \textsc{locked} state indicating
$\lockset(u)$ now represents the lock set $\mathcal{L}(u)$.
Note that taking $K = \Theta(\Delta^2)$ incurs messages of size $\Theta(\log \Delta)$, an increase from the $\Theta(1)$-message size in~\cite{Daymude2022-localmutual}, for $\Delta=\omega (1)$. This is necessary, however, to guarantee a polynomial bound on the number of open competition trials it takes in expectation for a node $u$ to win,
as already pointed out in~\cite{Daymude2022-localmutual}.

A participant $v$ is responsible for coordinating the competition among all initiators that want to lock $v$.
To delineate successive competitions, $v$ distinguishes among initiators that are \textit{candidates} in the current competition, \textit{applicants} that may join the current competition, and those that are \textit{on hold} for the next competition.
When $v$ receives a \msg{prepare}{} message from an initiator $u$, it either puts $u$ on hold (i.e., $H(v) \gets H(v) \cup \{\ell_v(u)\}$) if a competition is already underway or adds $u$ as an applicant (i.e., $A(v) \gets A(v) \cup \{\ell_v(u)\}$), advances to the \textsc{prepare} phase and replies \msg{ready}{} otherwise.
Participant $v$ promotes its applicants to candidates (i.e., $C(v) \gets C(v) \cup \{\ell_v(u)\}$) when $v$ receives their \msg{request-lock}{$p$} messages and advances to the \textsc{compete} phase.
Once all such messages are received from the competition's candidates, $v$ notifies the one with the unique highest priority of its success and all others of their failure (or, in the case of a tie, all candidates fail).
A winning competitor is removed from the candidate set while all others remain to try again; once the candidate set is empty, $v$ promotes all initiators that were on hold to applicants.
Finally, when $v$ receives a \msg{set-lock}{} message, it sets its \lock\ variable accordingly and responds with an \msg{ack-lock}{} message.

\section{Locking Time Analysis} \label{sec:analysis}

In this section, we bound the locking time of the (refined)
local mutual exclusion algorithm of Daymude, Richa, and Scheideler~\cite{Daymude2022-localmutual}, i.e., the number of rounds elapsed from the time a node issues its lock request to the time it successfully obtains locks on itself and its persistent neighbors.
We first characterize the set and maximum number of enabled action executions a node $u$ can have at any point in time (Lemmas~\ref{lem:initiatorcheckdisable}--\ref{lem:action_execution}), allowing us to bound the expected number of rounds for a continuously enabled action execution to either be executed or become disabled (Lemma~\ref{lem:progress}).

We then turn to the analysis of {\em competition trials}, i.e., the intervals in which an initiator sends a random priority to its participants, waits for those participants to receive priorities from all their competing initiators, and then is informed of whether it won (with the unique highest priority) over all its participants.
As is standard in mutual exclusion analyses, our locking time for an initiator $u$  ignores any round during which some participant $v$ of an initiator $u$ is already locked by another node $w \neq u$, prohibiting $u$ from making progress until $w$ unlocks $v$; we call such a round a {\em closed round} for node $u$. We call any round that is not closed an {\em open round} for node $u$ (Definition~\ref{def:open}).
Lemma~\ref{lem:competitiontrialtime} bounds the expected number of open rounds spent within a single competition trial and Theorem~\ref{thm:lockrounds} analyzes the number of competition trials a node needs to participate in before it will win one, together yielding the locking time bound of $\bigo{n\Delta^3}$ open rounds in expectation.

Throughout this analysis, we categorize actions as either \textsc{Receive} actions (e.g., \textsc{ReceivePrepare}$^P$, \textsc{ReceiveReady}$^I$) or \textsc{Check} actions (e.g., \textsc{CheckStart}$^I$, \textsc{CheckPriorities}$^P$); we also categorize them as either initiator actions or participant actions depending on the role of nodes that execute them.

\begin{lemma} \label{lem:initiatorcheckdisable}
    At most one initiator \textsc{Check} action is enabled per node per time; moreover, no enabled initiator \textsc{Check} action is ever disabled.
\end{lemma}
\begin{proof}
    Initiator \textsc{Check} actions' guards have two conditions: one requiring the node's \state\ to have a particular value (e.g., $\state = \textsc{prepare}$ in \textsc{CheckStart}) and another requiring equality between the set of neighbors the node expects to receive messages from and the set of neighbors it has actually received messages from so far, ignoring disconnections (e.g., $R \setminus D = L \setminus D$ in \textsc{CheckStart}).
    Each initiator \textsc{Check} action requires a different \state\ value, so at most one can be enabled for any given node at a time.
    Moreover, a node's \state\ variable can only be modified by that node's execution of an initiator \textsc{Check} action.
    So an enabled initiator \textsc{Check} action can't be disabled due to a \state\ change.
    Otherwise, for a node with an enabled initiator \textsc{Check} action, the received and expected message sets are equal (less disconnections).
    Once they are equal, there cannot be any additional messages to receive; moreover, any disconnections are ignored by both sets symmetrically.
    So once these two sets are equal less disconnections, they remain equal less disconnections and cannot cause the initiator \textsc{Check} action to be disabled.
\end{proof}

Recall from Section~\ref{subsec:model} that an enabled node has one or more enabled action executions $\{(\alpha, m)\}$---where $\alpha$ is an action enabled for $u$ and $m$ is any one of the messages in transit to $u$ enabling $\alpha$, if applicable---and when an enabled node is activated, the adversary selects one of its enabled action executions to execute uniformly at random.

\begin{lemma} \label{lem:action_execution}
    Any node has at most $2\Delta+4$ enabled action executions at any given time.
\end{lemma}
\begin{proof}
    Consider any node $u$.
    In our algorithm, there are two kinds of enabled action executions $(\alpha, m)$: those where $\alpha$ is a \textsc{Check} action (with no associated message $m$) and those where $\alpha$ is a \textsc{Receive} action and $m$ is a message in transit to $u$ that is enabling $\alpha$.
    By Lemma~\ref{lem:initiatorcheckdisable}, $u$ can have at most one initiator \textsc{Check} action enabled per time.
    Additionally, there is only one participant \textsc{Check} action, \textsc{CheckPriorities$^{P}$}.
    So $u$ can have at most two enabled action executions where $\alpha$ is a \textsc{Check} action at any given time.

    Observe that the number of messages in transit to $u$ upper bounds its current number of enabled action executions where $\alpha$ is a \textsc{Receive} action.
    Node $u$ has at most $\Delta$ neighbors, and Lemma~11 of~\cite{Daymude2022-localmutual} guarantees that there are at most two messages in transit between $u$ and any one of these neighbors at any time.
    Thus, including the messages $u$ sends to itself, $u$ can have at most $2(\Delta+1)$ enabled action executions where $\alpha$ is a \textsc{Receive} action per time.
    Combined, this yields a total of at most $2 + 2(\Delta + 1) = 2\Delta + 4$ enabled action executions.
\end{proof}

We now show that our fairness assumption (Section~\ref{subsec:model}) ensures every continuously enabled action execution will be executed in $\bigo{\Delta}$ rounds in expectation.

\begin{lemma} \label{lem:progress}
    Every enabled action execution is either executed or disabled within $2\Delta + 4$ rounds of becoming enabled, in expectation.
\end{lemma}
\begin{proof}
    Consider any enabled action execution $(\alpha, m)$ of any enabled node $u$.
    If $(\alpha, m)$ is disabled within $2\Delta + 4$ expected rounds, we are done; otherwise, $u$ is continuously enabled for as long as $(\alpha, m)$ remains enabled.
    By Lemma~\ref{lem:action_execution}, $(\alpha, m)$ is one of at most $2\Delta + 4$ action executions enabled for node $u$ at any time, so $(\alpha, m)$ has at least $1/(2\Delta + 4)$ probability of being executed each time $u$ is activated.

    Suppose $(\alpha, m)$ became enabled for $u$ in round $r$.
    In any round $s > r$ by which $(\alpha, m)$ has neither been executed nor disabled, we know that $u \in \mathcal{E}_s$.
    By definition, round $s$ will not end until $u$ has been activated at least once, and our fairness assumption ensures $u$ must eventually be activated.
    In the worst case, $u$ is only activated once per round, so $(\alpha, m)$ has at least a $1/(2\Delta + 4)$ probability of being chosen and executed per round.
    For Bernoulli trials with success probability $p$, the expected number of trials to obtain a success is $1/p$; therefore, we conclude that $(\alpha, m)$ is executed (or is disabled) by round $r + 2\Delta + 4$ in expectation.
\end{proof}

Using Lemma~\ref{lem:progress}, we will bound the expected number of rounds elapsed from when an initiator node $u$ begins a lock request in \textsc{InitLock}$^I$ to when it successfully obtains locks over its persistent neighbors in \textsc{CheckDone}$^I$.
The preparation phase is straightforward: within $2\Delta + 4$ rounds, in expectation, all persistent neighbors $v$ that $u$ sent \msg{prepare}{} messages to in \textsc{InitLock}$^I$ will execute \textsc{ReceivePrepare}$^P$ and either put $u$ on hold (if a competition over $v$ is ongoing) or add $u$ to their applicant sets and reply with a \msg{ready}{} message.
Node $u$ will receive such \msg{ready}{} responses by executing \textsc{ReceiveReady}$^I$ within another $2\Delta + 4$ rounds in expectation and---if $u$ was not put on hold by any of its persistent neighbors---begin a competition by executing \textsc{CheckStart}$^I$ within another $2\Delta + 4$ expected rounds.

It remains to analyze the competition phase where initiators generate random priorities, send them to their participants in \msg{request-lock}{} messages, wait to learn competition outcomes in the form of \msg{win}{} responses, and retry with a new random priority if they lose.
We refer to each of these attempts as a \textit{competition trial}:

\begin{definition}[Competition Trial] \label{def:competitiontrial}
    A \underbar{competition trial} of an initiator $u$ is an interval that starts from $u$ sending out \msg{request-lock}{} messages as a result of \textsc{CheckStart}$^I$ or \textsc{CheckWin}$^I$ action executions and ends when $u$ executes \textsc{CheckWin}$^I$.
\end{definition}

We first bound the runtime of one competition trial and later return to the question of how many competition trials an initiator has to participate in before it wins one, in expectation.
More specifically, we are only interested in those competition trials in which there is a non-zero probability that an initiator can win.
If some participant $v$ is already locked, no competition trial involving $v$ can progress until $v$ is unlocked, so we do not count these rounds towards the locking time.

\begin{definition}[Open Competition Trial and Open Round] \label{def:open}
    A competition trial of an initiator $u$ is \underbar{open} if and only if $\forall v \in \lockset(u)$, $v$ remains unlocked throughout the competition trial.
    All rounds elapsed during an open competition trial of an initiator $u$ are considered \underbar{open} for $u$. 
\end{definition}

To the runtime of an (open) competition trial, we make use of the directed acyclic graph (DAG) $\mathcal D_i$ defined in~\cite{Daymude2022-localmutual} that captures the dependencies among all nodes involved in concurrent locks at stage $i$.
An initiator node $u$ is \textit{competing} if and only if $\state(u) = \textsc{compete}$, i.e., if $u$ has executed \textsc{CheckStart}$^I$ but has not yet received all \msg{win}{} messages needed to execute \textsc{CheckWin}$^I$.
Dependencies between competing initiators and participants at the start of stage $i$ are represented as a directed bipartite graph $\mathcal{D}_i = (\mathcal{I}_i \cup \mathcal{P}_i, E_i)$ where $\mathcal{I}_i = \{u : \state_i(u) = \textsc{compete}\}$ is the set of competing initiators and $\mathcal{P}_i = \{u : \exists v \in \mathcal{I}_i \text{ s.t.\ } u \in \lockset_i(v)\}$ is the set of participants.
Some nodes belong to both partitions and their initiator and participant roles are considered distinct.
For nodes $u \in \mathcal{I}_i$ and $v \in \mathcal{P}_i \cap \lockset_i(u) \cap N_i[u]$, the directed edge $(u, v) \in E_i$ if and only if $v$ has not yet sent a \msg{win}{} message to $u$ in response to the latest \msg{request-lock}{} message from $u$; analogously, $(v, u) \in E_i$ if and only if $u$ has not yet sent a \msg{request-lock}{} message to $v$ in response to the latest \msg{win}{} message from $v$.

\begin{lemma} \label{lem:competitiontrialtime}
   Every competing initiator resolves its competition trial within at most $2k(10\Delta + 20)$ open rounds in expectation, where $k$ is the number of initiators along the longest path starting at node $u$ in DAG $\mathcal{D}_i$ at stage $i$. 
\end{lemma}
\begin{proof}
    We define the \textit{dependency component} $C_{i,u}$ of $u$ at the start of stage $i$ as the subgraph of $\mathcal{D}_i$ reachable from $u$.
    The \textit{height} of this dependency component, $h$ is the number of nodes in the longest path from $u$ to a leaf node (i.e., a node with out-degree zero).
    Any dependency component is finite and acyclic because $\mathcal{D}_i$ itself is finite and acyclic as shown in Lemma 4 of~\cite{Daymude2022-localmutual}.
    Additionally, due to the bipartite nature of $\mathcal{D}_i$, any path starting from $u$ in $C_{i,u}$ will alternate between initiator nodes and participant nodes.
    
    Argue by induction on $h \geq 1$ that any initiator (resp., participant) $u$ with dependency component $C_{i,u}$ of height $h$ takes at most $h(10\Delta + 20)$ open rounds in expectation to complete its next \textsc{CheckWin}$^I$ (resp., \textsc{CheckPriorities}$^P$) action.
    In the base case of $h = 1$, $C_{i,u}$ contains only the node $u$ which is either an initiator or a participant.
    If $u$ is an initiator, then since it is in the DAG $\mathcal{D}_i$, it has previously sent out \msg{request-lock}{} messages to its participants; moreover, since $u$ is a leaf in $\mathcal{D}_i$, all those participants must have already sent \msg{win}{} responses back to $u$.
    Thus, all corresponding action executions of \textsc{ReceiveWin}$^I$ are enabled for $u$ and by Lemma~\ref{lem:progress} must be executed or disabled within $2\Delta + 4$ additional rounds, in expectation.
    This in turn enables \textsc{CheckWin}$^I$ for $u$ which likewise by Lemma~\ref{lem:progress} is executed within an expected $2\Delta + 4$ additional rounds.
    Thus, if $u$ is an initiator, it executes its next \textsc{CheckWin}$^I$ action within $4\Delta + 8$ open rounds, in expectation.

    Otherwise, if $u$ is a participant, its presence in the DAG indicates only that some initiator is still competing for it.
    As a leaf in $\mathcal{D}_i$, $u$ is not waiting for any \msg{request-lock}{} messages from existing competing initiators, but it is possible that it is waiting for \msg{request-lock}{} messages from new initiators in its applicant set who have not yet become competing and joined the DAG.
    By Lemma~\ref{lem:progress}, all participants of these new initiators must have received their \msg{prepare}{} messages within $2\Delta + 4$ open rounds from stage $i$; these new initiators then receive the resulting \msg{ready}{} messages during \textsc{ReceiveReady}$^I$ and respond with \msg{request-lock}{} messages via \textsc{CheckStart}$^I$ within another $4\Delta + 8$ open rounds; and finally---by the same logic as for initiators---$u$ must receive all \msg{request-lock}{} messages via \textsc{ReceiveRequest}$^P$ and execute \textsc{CheckPriorities}$^P$ within another $4\Delta + 8$ rounds, totaling $10\Delta + 20$ open rounds in expectation.
    This concludes the base case.

    Now consider any node $u$ with dependency component $C_{i,u}$ whose height is $h > 1$, and suppose by induction that all nodes $v$ whose dependency components $C_{i,v}$ have heights $h' < h$ will complete their next \textsc{Check} action within $h'(10\Delta + 20)$ open rounds, in expectation.
    By the induction hypothesis (using height 1), all leaves of $C_{i,u}$ will execute their next \textsc{Check} actions within $10\Delta + 20$ open rounds, in expectation.
    Let $j > i$ be the first stage after which all leaves of $C_{i,u}$ have completed at least one \textsc{Check} action.
    We claim that $C_{j,u}$ is a subgraph of $C_{i,u}$ with a strictly smaller height.
    Supposing this claim holds, then the induction hypothesis guarantees that $u$ will complete its next \textsc{Check} action within at most $(h-1)(10\Delta + 20)$ open rounds of stage $j$ in expectation, which is $h(10\Delta + 20)$ open rounds from stage $i$, concluding the induction argument.
    Any dependency component rooted at an initiator has height $h \leq 2k$---since paths in the DAG alternate between initiators and participants---proving the lemma.

    It remains to prove the claim.
    Certainly any node $v$ in $C_{j,u}$ whose dependency path from $u$ has not changed since stage $i$ cannot have new out-edges in stage $j$, since this requires executing a \textsc{Check} action which would have reversed the edge between $v$ and its parent(s).
    On the other hand, any node $v$ in $C_{i,u}$ that did execute a \textsc{Check} action between stages $i$ and $j$ cannot be in $C_{j,u}$, since that \textsc{Check} action reversed the edge between $v$ and its parent(s), disrupting any path from $u$ to $v$, which cannot be reestablished without $u$ also executing a \textsc{Check} action.
    But since stage $j$ occurs immediately after some leaf $\ell$ of $C_{i,u}$ completes its first \textsc{Check} action and there was a path of dependencies from $u$ to $\ell$, $u$ has not executed its next \textsc{Check} action by the start of stage $j$.
    Thus, since no new out-edges are added to $C_{i,u}$ before stage $j$ and every leaf of $C_{i,u}$ has executed a \textsc{Check} action by stage $j$, $C_{j,u}$ is a subgraph of $C_{i,u}$ with strictly smaller height, as desired.
\end{proof}

The preceding lemma bounds the runtime of any one (open) competition trial; in our main theorem, we combine this with the expected number of trials any initiator has to participate in before winning (i.e., drawing the unique highest priority over all participants) to obtain an upper bound on the locking time.
We assume the set of priorities is size $K = \Theta(\Delta^2)$, requiring $\Theta(\Delta)$ memory per node and messages of size $\Theta(\log \Delta)$.

\begin{theorem} \label{thm:lockrounds}
    For $K=\Theta(\Delta^2)$, it takes $\bigo{n \Delta^3}=\bigo{n^4}$ open rounds in expectation for an initiator to successfully complete its \Lock\ operation, with $\Theta(\Delta)$ memory size per node and messages of size $\bigo{\log \Delta}$.
\end{theorem}
\begin{proof}
    We wish to calculate the number of open rounds it takes in expectation for an initiator $u$ to complete its \Lock\ operation, i.e., the time it takes to go from \textsc{InitLock$^{I}$} to \textsc{CheckDone$^{I}$}. Lemma~\ref{lem:competitiontrialtime} implies that it takes at most $2n (10\Delta + 20)$ expected number of rounds for a competing 
    initiator to resolve its competition trial (since $n$ is an upper bound on the number of initiators). 

    Apart from the four actions involved in a competition trial that we already account for in Lemma \ref{lem:competitiontrialtime}, there are seven additional actions (namely, \textsc{InitLock$^{I}$}, \textsc{ReceivePrepare}$^P$, \textsc{ReceiveReady}$^I$, \textsc{CheckStart}$^I$, \textsc{ReceiveSetLock}$^P$, \textsc{ReceiveAckLock}$^I$ and \textsc{CheckDone}$^I$) on the path from \textsc{InitLock$^{I}$} to \textsc{CheckDone$^{I}$} that will each take at most $2\Delta +4$ rounds in expectation to execute, to a total of at most $7(2\Delta + 4)$ rounds in expectation. It remains to bound the number 
    of trials a competing initiator requires to win the competition, which we denote by $X$. To calculate the expected value of $X$, we utilize the lower bound on the probability $p$ of a node winning a trial as calculated in Lemma 7 of~\cite{Daymude2022-localmutual}, namely $p \geq \frac{(1-1/K)^{2\Delta^{2}}}{2\Delta^{2}}$. 
    It follows that the expected number of competition trials $X$ is bounded by

    \begin{equation}
        E[X]=1/p\leq\frac{2\Delta^{2}}{(1-1/K)^{2\Delta^{2}}}
    \end{equation}

    Putting it altogether, the expected number of rounds to complete a \Lock\ operation for an initiator $u$ is upper bounded by
    \begin{align} 7(2\Delta + 4) + 2k(10\Delta + 20)\cdot E[X] &\leq 7(2\Delta + 4) + 2n(10\Delta + 20)\cdot \left( \frac{2\Delta^{2}}{(1-1/K)^{2\Delta^{2}}}\right) \\
    &= (2\Delta + 4)\left(7+10n\left(\frac{2\Delta^{2}}{(1-1/K)^{2\Delta^{2}}}\right)\right)\\
    &\leq (2\Delta + 4)\left(7+10n\left(\frac{2\Delta^{2}}{e^{-4\Delta^{2}/K}}\right)\right) 
    \label{eq:taylor}\\
    &= (2\Delta + 4)[7+20e^{4/c}n\Delta^{2}] = \bigo{n\Delta^3} 
    \label{eq:Kbound}
    \end{align}
    Equation~\ref{eq:taylor} follows from the Taylor series expansion of $e^{-2x}$, 
    by which we get $1-x \geq e^{-2x}$, for $x\in[0,1/2]$, while Equation~\ref{eq:Kbound} follows from picking $K=c\Delta^2$, where $c$ is positive constant. 

    The nodes only need memory proportional to the maximum number of ports $\Delta$ and to $\log K= \Theta(\log \Delta)=\bigo{\Delta}$, and messages exchanged in the algorithm are of size at most $\bigo{\log K}=\bigo{\log \Delta}$. 
\end{proof}

By definition, lockout freedom requires that every lock request must eventually succeed with probability 1. Let $X\geq 0$ be the random variable that indicates the number of rounds for an initiator node $u$ to complete its lock operation. Theorem~\ref{thm:lockrounds} states that $E[X]<\infty$, thus trivially implying that $X$ also has to be finite (with probability 1), matching the lockout freedom guarantee of~\cite{Daymude2022-localmutual}.

\begin{corollary} \label{cor:deadlockfreedom}
    The local mutual exclusion algorithm satisfies lockout freedom.
\end{corollary}

If $\Delta$ is a constant, which is the case in several applications of interest including those addressed in~\cite{Daymude2022-localmutual}, we have the following result.

\begin{corollary} \label{cor:tightbound}
    If $\Delta = \Theta(1)$, the algorithm takes $\bigo{n}$ open rounds in expectation for an initiator to successfully complete its \Lock\ operation. 
\end{corollary}

The proof of Lemma 12 in~\cite{Daymude2022-localmutual} gives us a construction to move from an asynchronous schedule $\mathcal{S}$ to a semi-synchronous schedule $\mathcal{S'}$.
This construction preserves the causal relationship between action executions while adding in filler time steps (stages) to make sure that each node executes only one action per stage.
Combining this construction with our Theorem~\ref{thm:lockrounds} extends our locking time bound to the asynchronous setting.  

\begin{corollary} \label{cor:asynclocktime}
    Under asynchronous concurrency, the algorithm satisfies mutual exclusion, lockout freedom, and has a locking time of $\bigo{n\Delta^3}$ open rounds in expectation.  
\end{corollary}

\section{Conclusion} \label{sec:conclude}

In this paper, we analyzed the locking time of a prior algorithm for a variant of local mutual exclusion~\cite{Daymude2022-localmutual} that enables anonymous, message-passing nodes to isolate concurrent actions involving their persistent neighborhoods despite dynamic network topology. Specifically, we proved that a minor refinement of this algorithm guarantees that any node issuing a lock request will obtain its locks within $\bigo{n\Delta^3}$ open rounds, in expectation, where $n$ is the number of nodes in the dynamic network, $\Delta$ is the maximum degree of the dynamic network, rounds that can be viewed as normalized to the execution time of the ``slowest'' node, and open rounds are those in which a node's neighbors are not already locked by some other node, prohibiting progress. We note that the minor modifications introduced to the algorithm in Section~\ref{sec:alg} and to the scheduler in Section~\ref{subsec:model} are necessary formalizations for runtime analysis, since without them, one can only show eventual termination of the algorithm in~\cite{Daymude2022-localmutual}.

The adversarial scheduler considered in this paper selects an arbitrary subset of enabled nodes at any stage, and each activated node randomly selects one of its enabled actions executions. An alternative that would yield similar expected times as the ones proven here would be to have each node keep a FIFO queue of its enabled action executions, where action executions are inserted into the queue, and executed, in the order they become enabled, and are removed from the queue if they get disabled.
Note that, while the choice of action execution now becomes deterministic, the runtime bound of the algorithm would still be probabilistic, given the use of randomization in the algorithm. Another alternative approach would be to change the guards of the actions to induce a set of "priorities" over the enabled action executions (guaranteeing that certain action executions would only be enabled after "higher priority" action executions are all disabled). However, this does not work within the structure of actions of the algorithm in~\cite{Daymude2022-localmutual}, since any fixed priority approach applied to the current algorithmic structure  leads to the "starvation" of some action executions and prevents progress. In future work, it would be interesting to establish a lower bound revealing whether this locking time bound is tight; if so, it remains an open question whether any algorithm exists for local mutual exclusion whose locking time is free of dependence on the size of the dynamic network.

\bibliographystyle{plainurl}
\bibliography{ref}

\end{document}